\documentclass[11pt]{article}

\usepackage{cite}
\usepackage{graphicx}
\usepackage{amsmath}
\usepackage{amsthm}
\usepackage{yhmath}

\usepackage{amsfonts}
\usepackage{amssymb}
\usepackage{fullpage}
\usepackage{geometry}
\usepackage{color}
\usepackage{latexsym}
\usepackage{algorithm}
\usepackage[noend]{algorithmic}
\usepackage{multirow}

\newcommand{\old}[1]{{}}

\newtheorem{theorem}{Theorem}[section]
\newtheorem{corollary}[theorem]{Corollary}
\newtheorem{lemma}[theorem]{Lemma}
\newtheorem{claim}[theorem]{Claim}

\title{Bottleneck Non-Crossing Matching in the Plane\thanks{Work by A.K. Abu-Affash was partially supported by a fellowship for doctoral students from the Planning \& Budgeting Committee of the Israel Council for Higher Education, and by a scholarship for advanced studies from the Israel Ministry of Science and Technology. Work by A.K. Abu-Affash and Y. Trabelsi was partially supported by the Lynn and William Frankel Center for Computer Sciences. Work by P. Carmi was partially supported by grant 2240-2100.6/2009 from the German Israeli Foundation for scientific research and development, and by grant 87212211 from the Israel Science Foundation.
Work by M. Katz was partially supported by grant 1045/10 from the Israel Science Foundation. Work by M. Katz and P. Carmi was partially supported by grant 2010074 from the United States -- Israel Binational Science Foundation.
}}
\author{A. Karim Abu-Affash \ \ \ Paz Carmi \ \ \ Matthew J. Katz \ \ \ Yohai Trabelsi
\\
\\
{\small Department of Computer Science, Ben-Gurion University, Israel} \\
{\small {\tt $\{$abuaffas,carmip,matya,yohayt$\}$@cs.bgu.ac.il}}}

\old{
\author{A. Karim Abu-Affash\thanks{Department of Computer Science, Ben-Gurion University of the Negev, Beer-Sheva 84105, Israel,
{\tt abuaffas@cs.bgu.ac.il}.} %Partially supported by the Lynn and William Frankel Center for Computer Sciences, by a fellowship for outstanding doctoral students from the Planning \& Budgeting Committee of the Israel Council for Higher Education, and by a scholarship for advanced studies from the Israel Ministry of Science and Technology.}
\and Paz Carmi\thanks{Department of Computer Science, Ben-Gurion University of the Negev, Beer-Sheva 84105, Israel, {\tt carmip@cs.bgu.ac.il}.} %Partially supported by a grant from the German-Israeli Foundation.}
\and Matthew J. Katz\thanks{Department of Computer Science, Ben-Gurion University of the Negev, Beer-Sheva 84105, Israel, {\tt matya@cs.bgu.ac.il}.} %Partially supported by grant 1045/10 from the Israel Science Foundation, and by the Israel Ministry of Industry, Trade and Labor (consortium CORNET).}
\and Yohai Trabelsi\thanks{Department of Computer Science, Ben-Gurion University of the Negev, Beer-Sheva 84105, Israel, {\tt yohayt@cs.bgu.ac.il}.}%Partially supported by the Lynn and William Frankel Center for Computer Sciences.}
}
}

\begin{document}
\maketitle

\begin{abstract}

Let $P$ be a set of $2n$ points in the plane, and let $M_{\rm C}$ (resp., $M_{\rm NC}$) denote a bottleneck matching (resp., a bottleneck non-crossing matching) of $P$. We study the problem of computing $M_{\rm NC}$.
We first prove that the problem is NP-hard and does not admit a PTAS. Then, we present an $O(n^{1.5}\log^{0.5} n)$-time algorithm
that computes a non-crossing matching $M$ of $P$, such that $bn(M) \le 2\sqrt{10} \cdot bn(M_{\rm NC})$,
where $bn(M)$ is the length of a longest edge in $M$. An interesting implication of our construction is that
$bn(M_{\rm NC})/bn(M_{\rm C}) \le 2\sqrt{10}$.
%Finally, we show that in some special cases, e.g., when the points of $P$ are in convex position,
%it is possible to compute $M_{\rm NC}$ in polynomial time.
\end{abstract}
\newpage
%%%%%%%%%%%%%%%%%%%%%%%%%%%%%%%%%%%%%%%%%%%% Introduction %%%%%%%%%%%%%%%%%%%%%%%%%%%%%%%%%%%%
%%%%%%%%%%%%%%%%%%%%%%%%%%%%%%%%%%%%%%%%%%%%%%%%%%%%%%%%%%%%%%%%%%%%%%%%%%%%%%%%%%%%%%%%%%%%%%
\section{Introduction}

Let $P$ be a set of $2n$ points in the plane. A {\em perfect matching} $M$ of $P$ is a perfect matching in the complete Euclidean graph induced by $P$. Let $bn(M)$ denote the length of a longest edge of $M$. A {\em bottleneck matching} $M_{\rm C}$ of $P$ is a perfect matching of $P$ that minimizes $bn(\cdot)$. A {\em non-crossing matching} of $P$ is a perfect matching whose edges are pairwise disjoint.
In this paper, we study the problem of computing a bottleneck non-crossing matching of $P$; that is, a non-crossing matching $M_{\rm NC}$ of $P$ that minimizes $bn(\cdot)$, where only non-crossing matchings of $P$ are being considered.

The non-crossing requirement is quite natural, and indeed many researches have considered geometric problems dealing with crossing-free configurations in the plane; see, e.g.~\cite{Aichholzer09,Aichholzer08,Alon93,Aloupis10,Rappaport02}.
In particular, (bottleneck) non-crossing matching is especially important in the context of layout of VLSI circuits~\cite{Lengauer90} and operations research.
It is easy to see that there always exists a non-crossing matching of $P$ (e.g., match each point with the first point to its right). Actually, any minimum weight matching of $P$ is non-crossing. However, as shown in Figure~\ref{fig:int1}, which is borrowed from~\cite{Carlsson10}, the length of a longest edge of a minimum weight matching can be much larger than that of a bottleneck non-crossing matching.

\begin{figure}[htp]
    \centering
        \includegraphics[width=.9\textwidth]{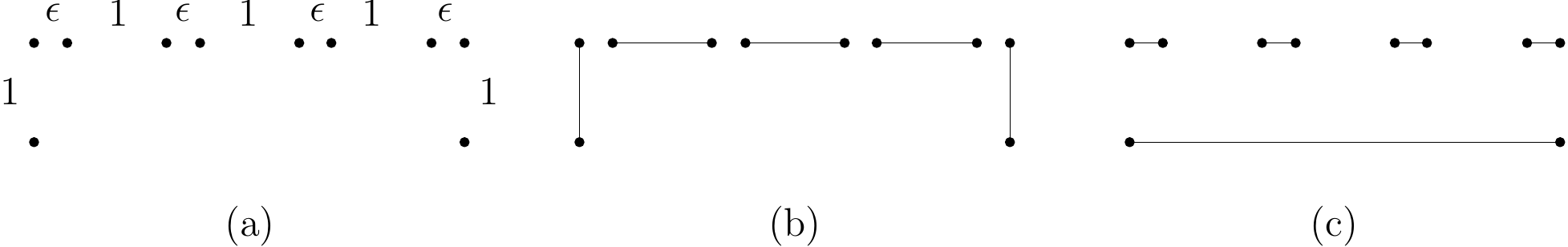}
    \caption{(a) The set $P$. (b) A bottleneck non-crossing matching $M_1$ of $P$. (c) A minimum weight matching $M_2$ of $P$. Notice that $bn(M_2)/bn(M_1) \rightarrow n-2$.}
    \label{fig:int1}
\end{figure}

\subsection{Related work}
Matching problems play an important role in graph theory, and thus have been studied extensively, see~\cite{Lovasz86}.
The various matching algorithms developed for general weighted graphs of course apply in our setting. However, it turns out that one can do better in the case of points in the plane. Vaidya~\cite{Vaidya88} presented an $O(n^{5/2}\log^4 n)$-time algorithm for computing a minimum weight matching, based on Edmonds' $O(n^3)$ algorithm. Subsequently, Varadarajan~\cite{Vara98} described an $O(n^{3/2}\log^5 n)$-time algorithm for this problem. For the bipartite version, Vaidya~\cite{Vaidya88} presented an $O(n^{5/2}\log n)$-time algorithm and Agarwal et al.~\cite{AES99} presented an $O(n^{2+\varepsilon})$-time algorithm; both algorithms are based on the Hungarian method~\cite{Lovasz86}. As for bottleneck matching, Chang et al.~\cite{CTL92} obtained an $O(n^{3/2}\log^{1/2} n)$-time algorithm for computing a bottleneck matching, by proving that such a matching is contained in the 17RNG. Efrat and Katz~\cite{EK00} extended this result to higher dimensions. For the bipartite version, Efrat et al.~\cite{Efrat01} presented an $O(n^{3/2} \log n)$-time algorithm.
Algorithms for other kinds of matchings, as well as approximation algorithms for the problems above, have also been developed.

Self-crossing configurations are often undesirable and might even imply an error condition; for example, a potential collision between moving objects, or inconsistency in the layout of a circuit. Many of the structures studied in computational geometry are non-crossing,
for instance, minimum spanning tree, minimum weight matching, Voronoi diagram, etc. Jansen and Woeginger~\cite{JW93} proved that
deciding whether there exists a non-crossing matching of a set of points with integer coordinates, such that all edges are of length {\em exactly} $d$, for a given integer $d \ge 2$, is NP-complete. Carlsson and Armbruster~\cite{Carlsson10} proved that the bipartite version of the bottleneck non-crossing matching problem is NP-hard. Alon et al.~\cite{Alon93} considered the problem
of computing the longest (i.e., maximum weight) non-crossing matching of a set of points in the plane. They presented an approximation algorithm
that computes a non-crossing matching of length at least $2/\pi$ of the length of the longest non-crossing matching.
Aloupis et al.~\cite{Aloupis10} considered the problem of finding a non-crossing matching between points and geometric objects in the plane.
See also~\cite{Aichholzer09, Aichholzer08, Rappaport02} for results related to non-crossing matching.

\subsection{Our results}
We begin by proving (in Section~2) that the problem of computing $M_{\rm NC}$ is NP-hard. Our proof is based on a reduction from the planar 3-SAT problem, and is influenced by the proof of Carlsson and Armbruster mentioned above. As a corollary we obtain that the problem does not admit a PTAS.
Next, in Section~3, we present an algorithm for converting any (crossing) matching $M_\times$ into a non-crossing matching $M_=$, such that
$bn(M_=) \le 2\sqrt{10} \cdot bn(M_\times)$. The algorithm consists of two stages: converting $M_\times$ into an intermediate (crossing) matching $M'_\times$ with some desirable properties, and using $M'_\times$ as a ``template'' for the construction of $M_=$.
The algorithm implies that (i) $M_{\rm NC}/M_{\rm C} \le 2\sqrt{10}$, and (ii) one can compute, in $O(n^{3/2}\log^{1/2} n)$-time, a non-crossing matching $M$, such that $bn(M) \le 2\sqrt{10} \cdot bn(M_{\rm NC})$. We are not aware of any previous constant-factor approximation algorithm for the problem of computing $M_{\rm NC}$. In the full version of this paper, we also present an $O(n^3)$-algorithm, based on dynamic programming, for computing $M_{\rm NC}$ when the points of $P$ are in convex position.

%%%%%%%%%%%%%%%%%%%%%%%%%%%%%%%%%%%%%%%%%%%% Section 2 %%%%%%%%%%%%%%%%%%%%%%%%%%%%%%%%%%%%%%%
%%%%%%%%%%%%%%%%%%%%%%%%%%%%%%%%%%%%%%%%%%%%%%%%%%%%%%%%%%%%%%%%%%%%%%%%%%%%%%%%%%%%%%%%%%%%%%
\section{Hardness Proof}\label{sec:Sec2}

In this section, we prove the following theorem.
Our proof is influenced by the proof of Carlsson and Armbruster for the bipartite version~\cite{Carlsson10}.
%from which we were influenced.

\begin{theorem}
Let $P$ be a set of $2n$ points in the plane. Then, computing a bottleneck non-crossing matching of $P$ is NP-hard.
\end{theorem}

\begin{proof}[\bf{\textit{Proof:}}]
The proof is based on a reduction from the planar 3-SAT problem.
%and is essentially inspired by a similar construction used in~\cite{Carlsson10}.
Given a 3-CNF formula $F$ with $n$ variables $X=\{x_1,x_2,\ldots,x_n\}$ and $m$ clauses $Y=\{C_1,C_2,\ldots,C_m\}$,
let $G_F=(V,E)$ be the graph of $F$, i.e., $V=X \cup Y$ and $E=\{(x_i,C_j):x_i \mbox{ appears in } C_j \mbox{ (either negated or unnegated)}\}$.
%Notice that the degree of a vertex in $Y$ is at most 3.
If $G_F$ is planar,
then $F$ is called a planar 3-CNF formula. The planar 3-SAT problem is to determine whether a given planar 3-CNF formula $F$ is satisfiable; the problem is NP-complete~\cite{Lichtenstein82}.

Let $F$ be a planar 3-SAT formula. We construct, in polynomial time, a set $P$ of points in the plane, such that $F$ is satisfiable if and only if there exists a non-crossing matching of $P$ with bottleneck 1.
Consider the graph $G_F$. It is well known that $G_F$ can be embedded in the plane in polynomial time.
%and logarithmic space (see~\cite{Allender05}).

\begin{figure}[htp]
    \centering
        \includegraphics[width=.5\textwidth]{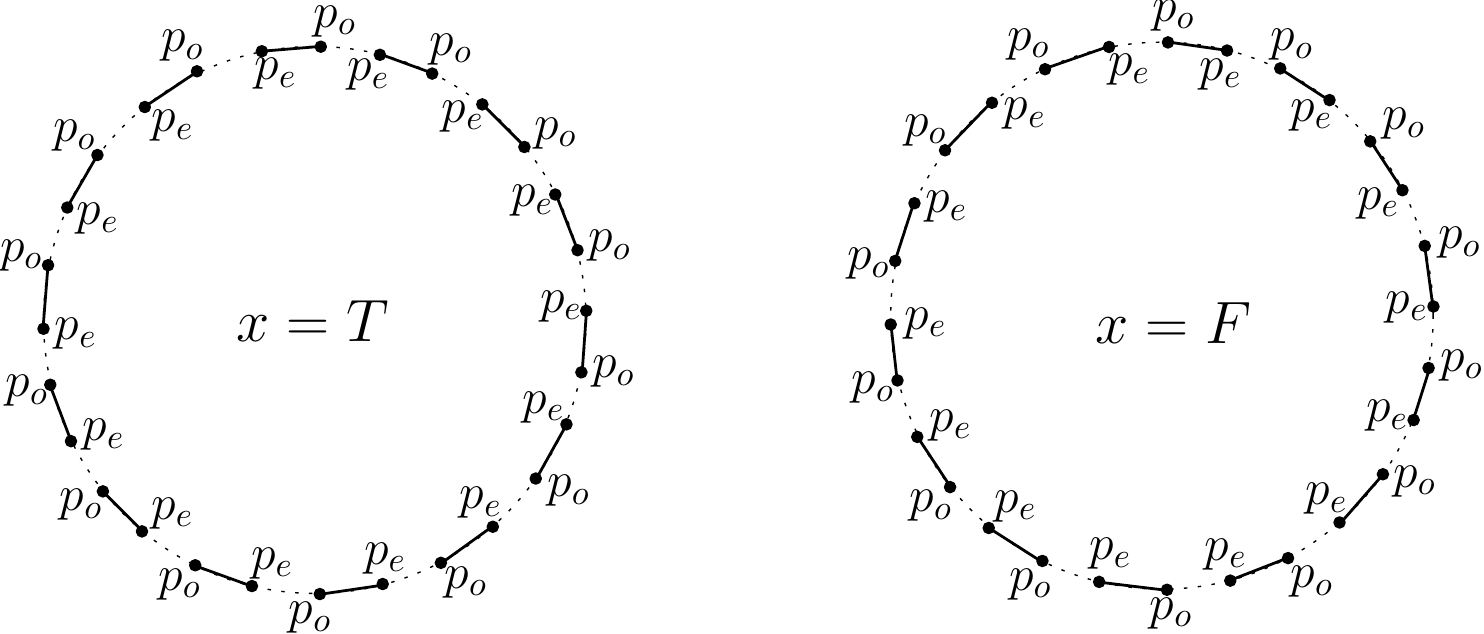}
    \caption{The true and the false matchings of a variable $x$.}
    \label{fig:reduction-1}
\end{figure}

{\bf Variables.} Each variable $x_i\in X$ is mapped to a circuit $s_i$ of an even number of points, where the distance between any two adjacent points is 1. We mark the points of $s_i$ alternately by $p_o$ and $p_e$. Each circuit $s_i$ can be partitioned into pairs of adjacent points in two ways. We arbitrarily associate one of them with the assignment $x_i=T$ and call it the ``true matching'', and the other with the assignment $x_i=F$ and call it the ``false matching''; see Figure~\ref{fig:reduction-1}. Thus, the value of $x_i$ will determine the matching on $s_i$, and vice versa.

\begin{figure}[htp]
    \centering
        \includegraphics[width=.23\textwidth]{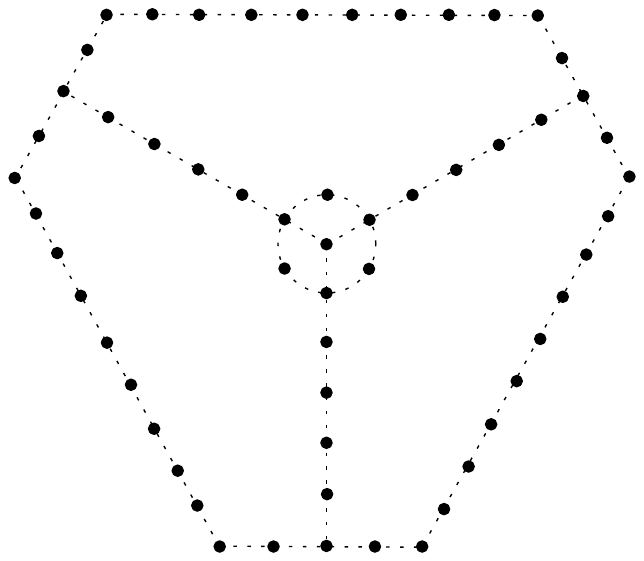}
    \caption{The hexagon $h_j$ corresponding to the clause $C_j$.}
    \label{fig:reduction-2}
\end{figure}

{\bf Clauses.} Each clause $C_j \in Y$ is mapped to a hexagonal component $h_j$ of points, where the distance between any two adjacent points is 1, as shown in Figure~\ref{fig:reduction-2}.

{\bf Edges.} Each edge $(x_i,C_j)$ between a variable $x_i$ and a clause $C_j$ is mapped to a path $l_{i,j}$ of an even number of points, where the distance between any two adjacent points is 1, that begins and ends at two different points of $s_i$ and intersects the hexagon $h_j$, as described below.

\begin{figure}[htp]
    \centering
        \includegraphics[width=.80\textwidth]{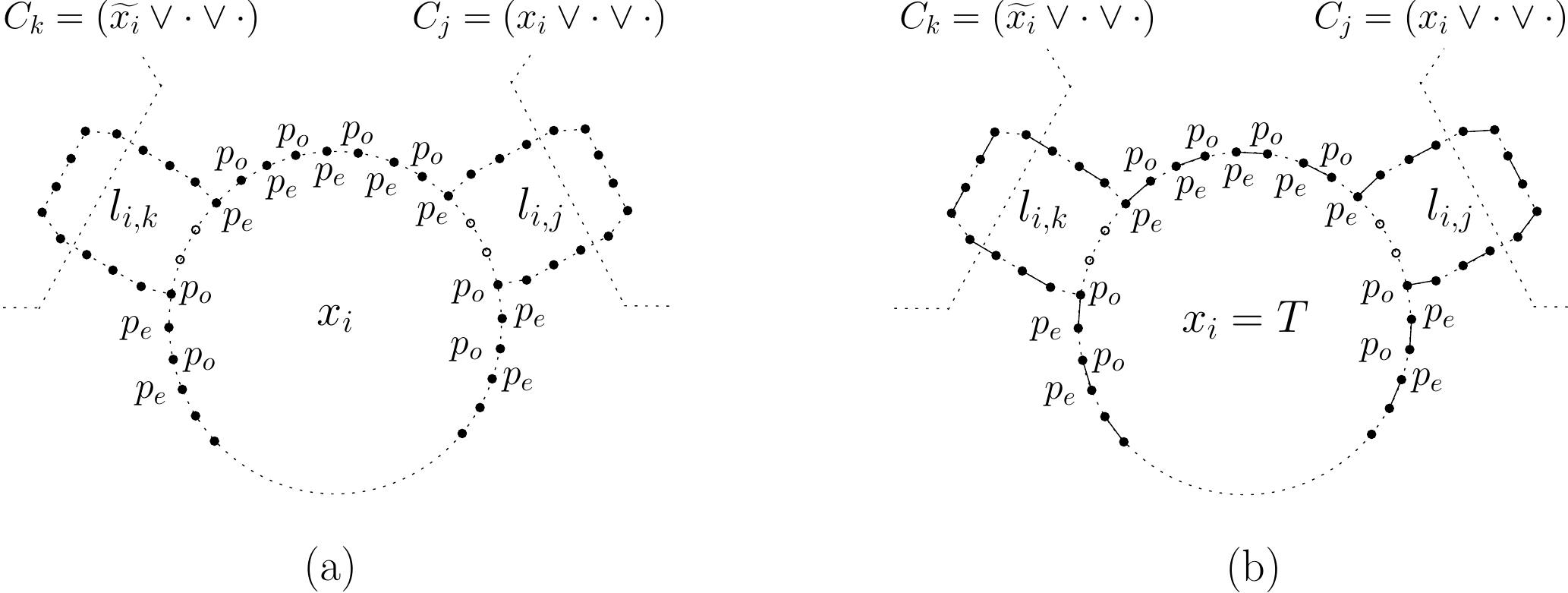}
    \caption{The paths between $x_i$ and the clauses $C_j$ and $C_k$. Since $x_i$ appears unnegated in $C_j$ and negated in $C_k$, $l_{i,j}$ begins at $p_e$ and ends at $p_o$, and $l_{i,k}$ begins at $p_o$ and ends at $p_e$.}
    \label{fig:reduction-3}
\end{figure}

Assume that $x_i$ appears unnegated in clause $C_j$ and negated in clause $C_k$, i.e., $C_j=(x_i\vee\cdot\vee\cdot)$ and $C_k=(\widetilde{x_i}\vee\cdot\vee\cdot)$. Consider the points of $s_i$ in clockwise order. For the unnegated instance, we select 4 consecutive points $p_e,p_o,p_e,p_o$ of $s_i$, connect $l_{i,j}$ to the first and last of these points, and remove the middle two points; see Figure~\ref{fig:reduction-3}(a). And, for the negated instance, we select 4 consecutive points $p_o,p_e,p_o,p_e$, connect $l_{i,k}$ to the first and last of these points, and remove the middle two points.
Notice that, as for circuits, a path can be partitioned into pairs of adjacent points in two ways. The value of $x_i$ and whether or not it appears negated in the clause will determine which of these two matchings will be the matching on the path, see below.

\begin{figure}[htp]
    \centering
        \includegraphics[width=.45\textwidth]{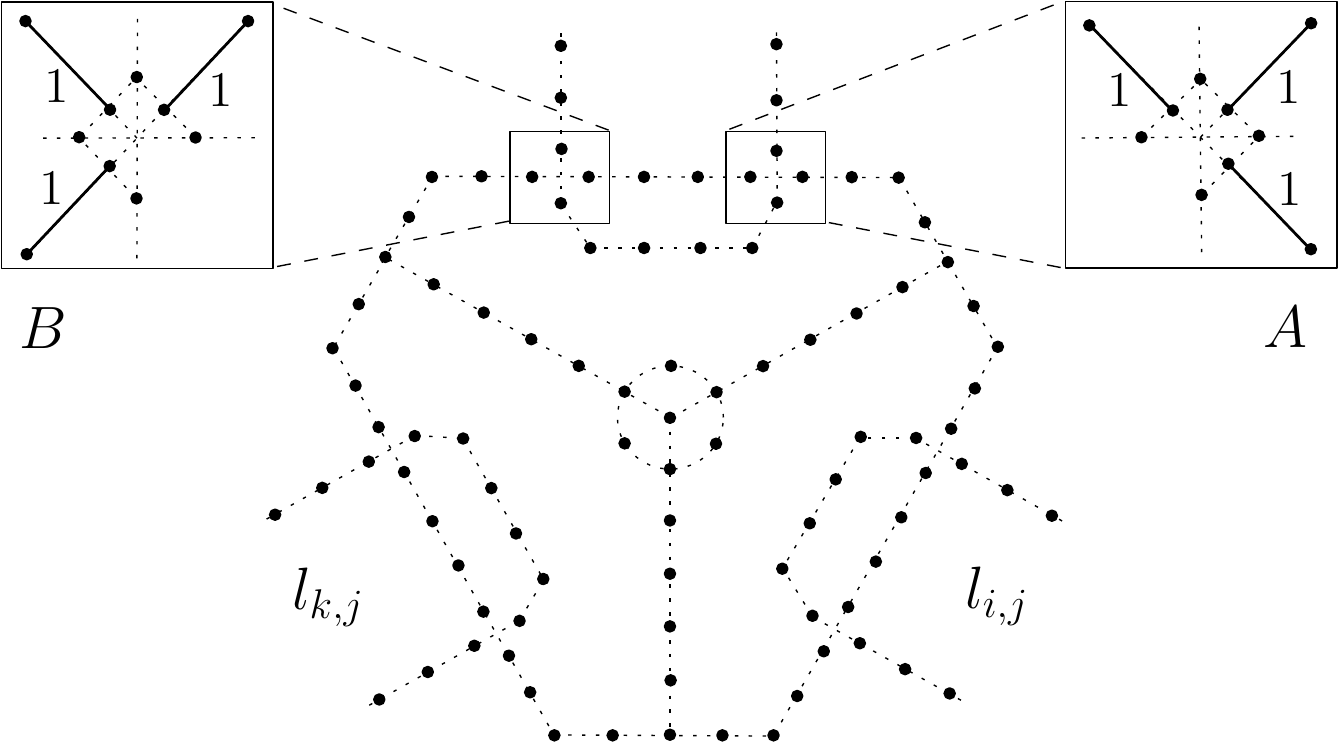}
    \caption{The intersection between $C_j$ and the variables appearing in it.}
    \label{fig:reduction-4}
\end{figure}

The intersection between $l_{i,j}$ and the hexagon $h_j$ representing $C_j$ is shown in Figure~\ref{fig:reduction-4}. $l_{i,j}$ intersects $h_j$ such that six of its points lie inside $h_j$ and the other lie outside $h_j$.
We would like to prevent situations where a point of $l_{i,j}$ is matched to a point of $h_j$.
%since we want the matchings on the paths intersecting $h_j$ to determine the matching on $h_j$.
To this end, we add 3 pairs of points around each of the two junctions involving $l_{i,j}$ and $h_j$, see squares $A$ and $B$ in Figure~\ref{fig:reduction-4}. This forces the points of $l_{i,j}$ to be matched to each other.

Notice that in the true (resp., false) matching of $s_i$, if $x_i$ appears unnegated (resp., negated) in $C_j$, then the six points of $l_{i,j}$ in $h_j$ are matched to each other (see Figure~\ref{fig:reduction-3}(b)), and, if $x_i$ appears negated (resp., unnegated) in $C_j$, then the two extreme points among these six points are matched to points outside $h_j$.

\begin{figure}[htp]
    \centering
        \includegraphics[width=.70\textwidth]{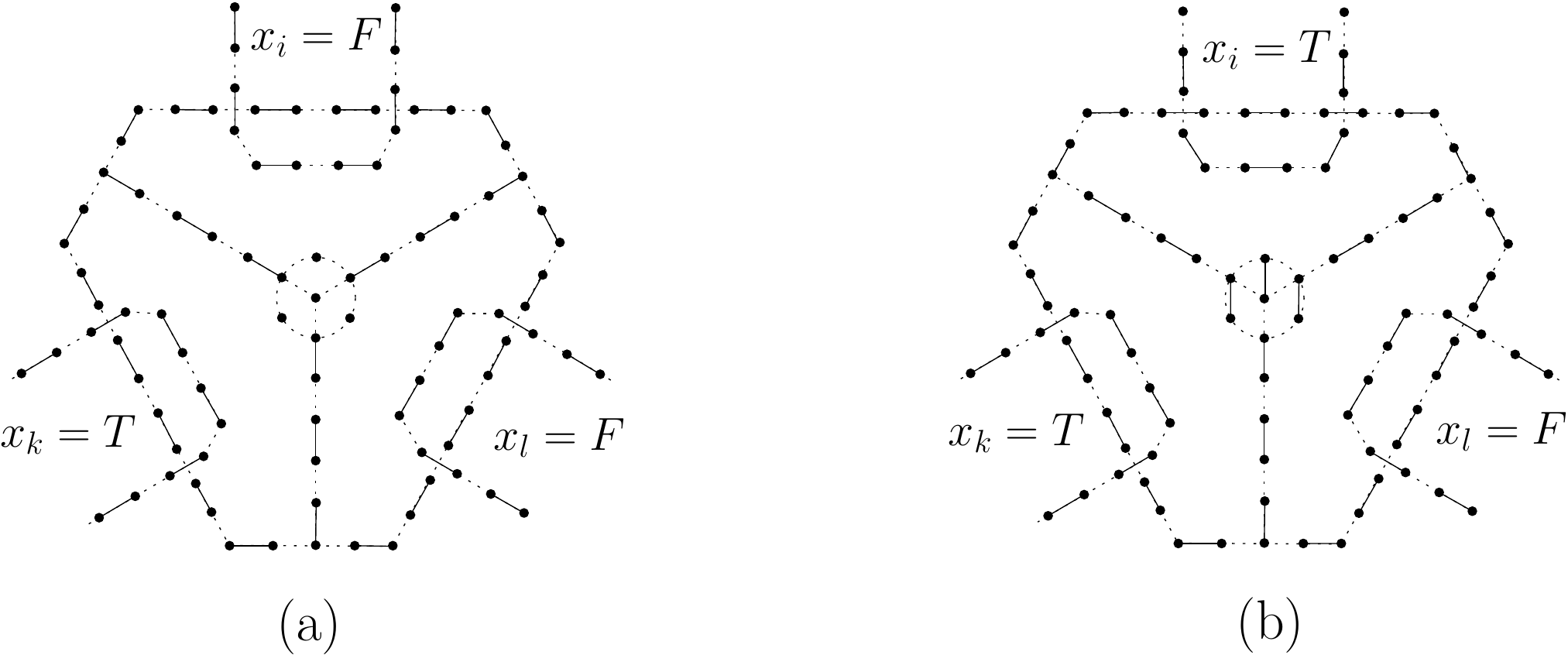}
    \caption{$C_j=(x_i\vee \widetilde{x_k}\vee x_l)$. (a) Assigning values to the variables $x_i,x_k,x_l$ so that the value of each of the corresponding literals is false, ensures that a matching on $h_j$ with bottleneck 1 does not exist. (b) Assigning value to the variables so that the value of at least one of the literals is true, ensures that such a matching does exist.}
    \label{fig:reduction-5}
\end{figure}

To see the correctness of the reduction, consider Figure~\ref{fig:reduction-5}.
%First observe that a matching on $h_j$ with bottleneck 1 exists if and only if for at least one of the paths intersecting $h_j$,
%its six points lying in $h_j$ are matched to each other.
First observe that there exists a matching on $h_j$ with bottleneck 1 if, and only if, for at least one of the three paths intersecting $h_j$, its six points lying in $h_j$ are matched to each other (i.e., the value of the corresponding literal is true).
Indeed, in Figure~\ref{fig:reduction-5}(a), the value of each of the three literals of $C_j$ is false, ensuring that a matching on $h_j$ with bottleneck 1 does not exist. And, in Figure~\ref{fig:reduction-5}(b), the value of one of the literals ($x_i$) of $C_j$ is true, ensuring that such a matching does exist. Notice that if the value of $\widetilde{x_k}$ and/or $x_l$ were also true, then one could still use the same matching on $h_j$.

Conversely, assume there exists a matching with bottleneck 1 and consider the truth assignment implied by the matchings on the circuits.
We need to verify that each of the clauses is satisfied by this truth assignment.
Let $h_j$ be the hexagon representing $C_j$. Since there exists a matching on $h_j$ with bottleneck 1, then, for at least one of the paths intersecting $h_j$, its six internal points are necessarily matched to each other. But this implies, as noted above, that the value of the corresponding literal, and therefore of $C_j$, is true.

\end{proof}

Finally, we observe that if one can find a non-crossing matching of bottleneck less than $3\sqrt{2}/4$, then one can solve the planar 3-SAT problem in polynomial time. This bound is obtained from examining the additional points that are added to each hexagon (see squares $A$ and $B$ in Figure~\ref{fig:reduction-4}).

\begin{corollary}
The problem of computing a bottleneck non-crossing matching does not admit a PTAS, unless P=NP.
\end{corollary}

%%%%%%%%%%%%%%%%%%%%%%%%%%%%%%%%%%%%%% Section 3 %%%%%%%%%%%%%%%%%%%%%%%%%%%%%%%%%%%%%%%%%%%
%%%%%%%%%%%%%%%%%%%%%%%%%%%%%%%%%%%%%%%%%%%%%%%%%%%%%%%%%%%%%%%%%%%%%%%%%%%%%%%%%%%%%%%%%%%%

\section{Approximation Algorithm}

Let $P$ be a set of $2n$ points in general position in the plane.
The {\em bottleneck} of a perfect matching $M$ of $P$, denoted $bn(M)$, is the length of a longest edge of $M$.
Let $M_\times$ be a perfect matching of $P$. In this section we show how to convert $M_\times$ into
a non-crossing perfect matching $M_=$ of $P$, such that $bn(M_=) \le 2\sqrt{10} \cdot bn(M_\times)$.

Set $\delta = bn(M_\times)$.
We begin by laying a grid of edge length $2\sqrt{2}\delta$. W.l.o.g. assume that each of the points in $P$ lies in the interior of some grid cell.
Consider an edge $e$ of $M_\times$. Since $e$ is of length at most $\delta$, it is either contained in a single grid cell, or
its endpoints lie in two adjacent cells (i.e., in two cells sharing a side or only a corner). In the former case, we say that $e$ is {\em internal}, and in the latter case, we say that $e$ is {\em external}. We distinguish between two types of external edges: a {\em straight} external edge (or s-edge for short) connects between a pair of points in two cells that share a side, while a {\em diagonal} external edge (or d-edge for short) connects between a pair of points in two cells that share only a corner.
Finally, the {\em degree} of a grid cell $C$, denoted $deg(C)$, is the number of external edges with an endpoint in $C$.

Our algorithm consists of two stages. In the first stage we convert $M_\times$ into another perfect matching $M'_\times$, such that (i) each edge of $M'_\times$ is either contained in a single cell, or connects between a pair of points in two adjacent cells, (ii) for each grid cell $C$, $deg(C) \le 4$ and these $deg(C)$ edges connect $C$ to $deg(C)$ of its adjacent cells, and (iii) some additional properties hold (see below).
In the second stage, we construct the matching $M_=$ according to $M'_\times$, such that
%, as for $M'_\times$, each edge of $M_=$ is either contained in a single cell, or connects between a pair of points
%in two adjacent cells. Moreover,
there is a one-to-one correspondence between the external edges of $M'_\times$ and the external edges of $M_=$. That is, there exists an edge in $M'_\times$ connecting between two adjacent cells $C_1$ and $C_2$ if and only if there exists such an edge in $M_=$. However, the endpoints of an external edge of $M_=$ might be different than those of the corresponding edge of $M'_\times$.

The second stage itself consists of two parts. In the first part, we consider each non-empty grid cell separately. When considering such a cell $C$, we first determine the at most four points that will serve as endpoints of the external edges of $C$ (as dictated by $M'_\times$).
%That is, if, e.g., there are two external edges in $M'_\times$ with endpoints in $C$, one connecting $C$ to a point in the cell above $C$
%and the other connecting $C$ to a point in the cell to its right, then we find a point $p_u$ in $C$ that will eventually connect
%to a point in the cell above $C$, and a point $p_r$ in $C$ that will eventually connect to a point in the cell to the right of $C$.
Next, we construct a non-crossing matching for the remaining points in $C$. In the second part of this stage, we add the external edges between the points that were chosen as endpoints for these edges in the first part.

We now describe each of the stages in detail.
\subsection{Stage~1}

In this stage we convert $M_\times$ into $M'_\times$.
We do it by applying a sequence of reduction rules to the current matching, starting with $M_\times$. Each of the rules is applied multiple times, as long as there is an instance in the current matching to which it can be applied. When there are no more such instances, we move to the next rule in the sequence.

\begin{figure}[htb]
    \centering
        \includegraphics[width=0.55\textwidth]{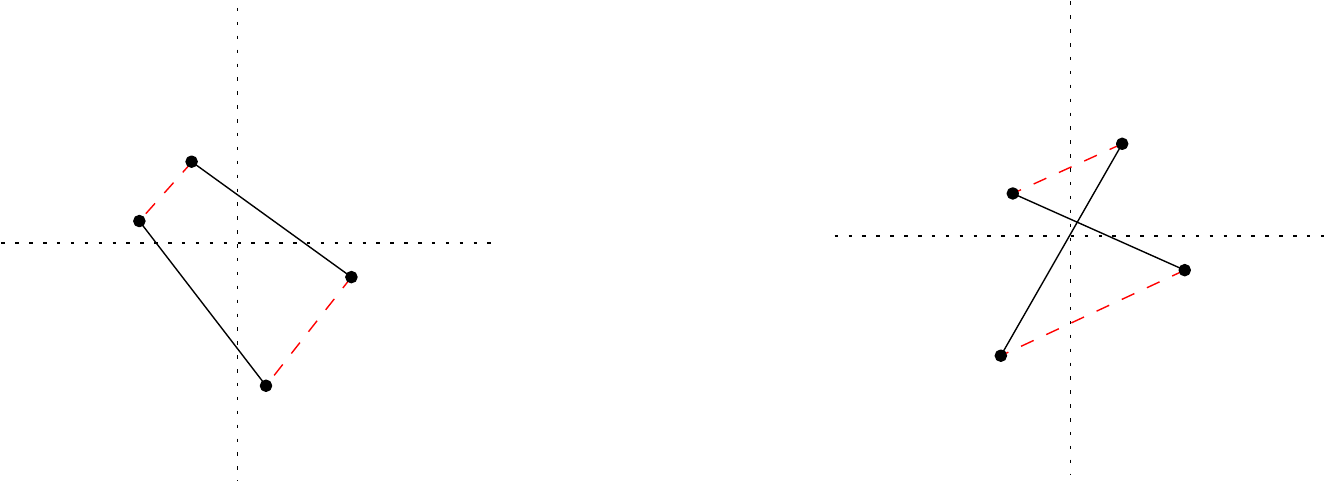}
        \caption{Rule I (left) and Rule II (right). The solid edges are replaced by the dashed edges.}
    \label{fig:rules_I_and_II}
\end{figure}

We associate a d-edge connecting between two cells with the corner shared by these cells.\\
{\bf Rule~I} is applied to a pair of d-edges associated with the same corner and connecting between the same pair of cells, see Figure~\ref{fig:rules_I_and_II}. The d-edges are replaced by a pair of internal edges.\\
{\bf Rule~II} is applied to a pair of d-edges associated with the same corner and connecting between different pairs of cells, see Figure~\ref{fig:rules_I_and_II}. The d-edges are replaced by a pair of s-edges.\\
Notice that when we are done with Rule~I, each corner has at most two d-edges associated with it, and if it has two, then they connect between different pairs of cells. Moreover, when we are done with Rule~II, each corner has at most one d-edge associated with it.
%since the corners that had two d-edges associated with them (when we were done with Rule~I), have no d-edge associated with them now.
Finally, since the length of a d-edge is at most $\delta$, any edge created by Rule~I or Rule~II is contained in the disk $D_\delta(a)$ of radius $\delta$ centered at the appropriate corner $a$.

%\begin{figure}[htp]
%    \centering
%        \includegraphics[width=0.4\textwidth]{danger_zone}
%        \caption{The danger zones defined by a d-edge.}
%    \label{fig:danger_zone}
%\end{figure}

\begin{figure}[htp]
    \centering
        \includegraphics[width=0.55\textwidth]{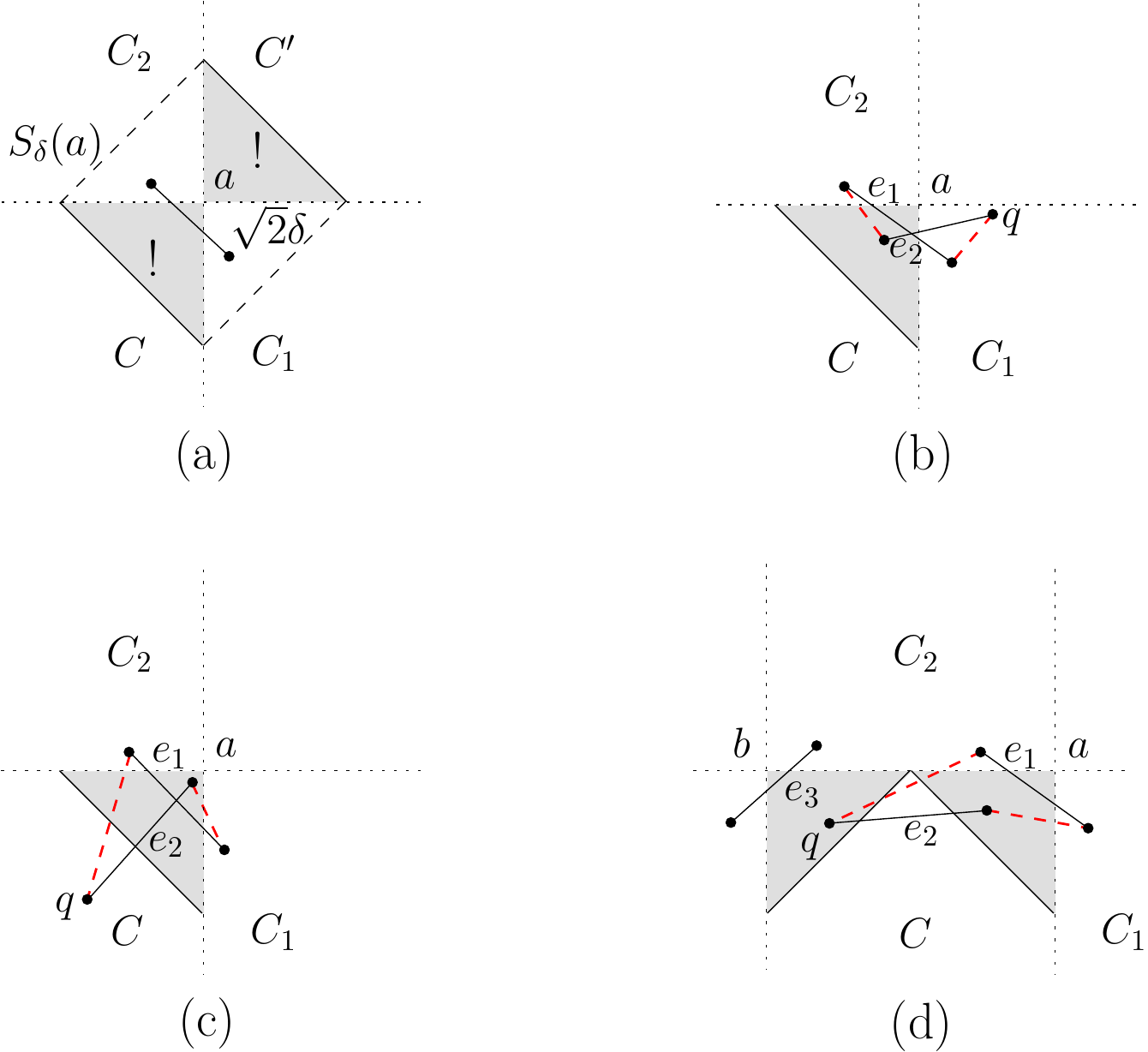}
        \caption{(a) The danger zones defined by a d-edge. Rule III: (b) $e_2$ is a s-edge. (c) $e_2$ is an internal edge and $q$ is not in another danger zone. (d) $e_2$ is an internal edge and $q$ is in another danger zone.}
    \label{fig:ruleIII}
\end{figure}

A d-edge associated with a corner $a$ defines a triangular {\em danger zone} in each of the two other cells sharing $a$, see Figure~\ref{fig:ruleIII}(a). A danger zone is semi-open; it does not include the hypotenuse. Let $S_\delta(a)$ denote the square that two of its sides are the hypotenuses of the danger zones defined by the d-edge associated with $a$. Notice that if $p$ is a point in cell $C$ that is not in the danger zone in $C$, then one cannot draw a d-edge between $C_1$ and $C_2$, with endpoints in the interior of $S_\delta(a)$, that passes through $p$.\\
{\bf Rule~III} is applied to a d-edge $e_1$ and an edge $e_2$ with an endpoint in a danger zone defined by $e_1$, see Figure~\ref{fig:ruleIII}(b--d).
We distinguish between two cases. If $e_2$ is a s-edge, then, by the claim below, its other endpoint is in one of the cells $C_1$ or $C_2$. In this case, we replace $e_1$ and $e_2$ with an internal edge and a s-edge, see Figure~\ref{fig:ruleIII}(b).
If $e_2$ is an internal edge, then consider the other endpoint $q$ of $e_2$. If $q$ is not in a danger zone in $C$ defined by another d-edge, then replace $e_1$ and $e_2$ with two s-edges, see Figure~\ref{fig:ruleIII}(c). If, however, $q$ is in a danger zone in $C$ defined by another d-edge $e_3$, then, by the claim below, $e_3$ is associated with one of the two corners of $C$ adjacent to the corner $a$ (to which $e_1$ is associated), and therefore, either $C_1$ or $C_2$ contains an endpoint of both $e_1$ and $e_3$.
Replace $e_1$ and $e_2$ with two s-edges, such that one of them connects $q$ to the endpoint of $e_1$ that is in the cell that also contains an endpoint of $e_3$, see Figure~\ref{fig:ruleIII}(d).

\begin{claim}
Consider Figure~\ref{fig:ruleIII}.
In an application of Rule~III, (i) if $e_2$ is a s-edge, then its other endpoint $q$ is in one of the cells $C_1$ or $C_2$, and (ii) if $q$ is in a danger zone defined by another d-edge $e_3$ associated with corner $b$, then $\overline{ab}$ is a side of $C$ and the danger zone containing $q$ is on the side of $b$ containing $a$.
\end{claim}
\begin{proof}
Statement~(i) is surely true just before the first application of Rule~III, since the length of $e_2$ then is at most $\delta$. (Notice that a s-edge created by Rule~II cannot have an endpoint in a danger zone.) For the same reason, Statement~(ii) is surely true just before the first application of Rule~III. (Notice that if $e_2$ is an internal edge created by Rule~I, then it is contained in one of the two danger zones defined by $e_1$.) It remains to verify that if $e_2$ was created by a previous application of Rule~III, then both statements are still true. Indeed, if $e_2$ is a s-edge created by a previous application of Rule~III, then it had to be an application of the type depicted in Figure~\ref{fig:ruleIII}(d), and the replacement instructions for this type ensure that Statement~(i) is true. As for Statement~(ii), if $e_2$ was created by an application of Rule~III, then, since $q$ was an endpoint of a d-edge that was removed by the application of Rule~III, it is not in a danger zone.
\end{proof}

\begin{figure}[htp]
    \centering
        \includegraphics[width=0.55\textwidth]{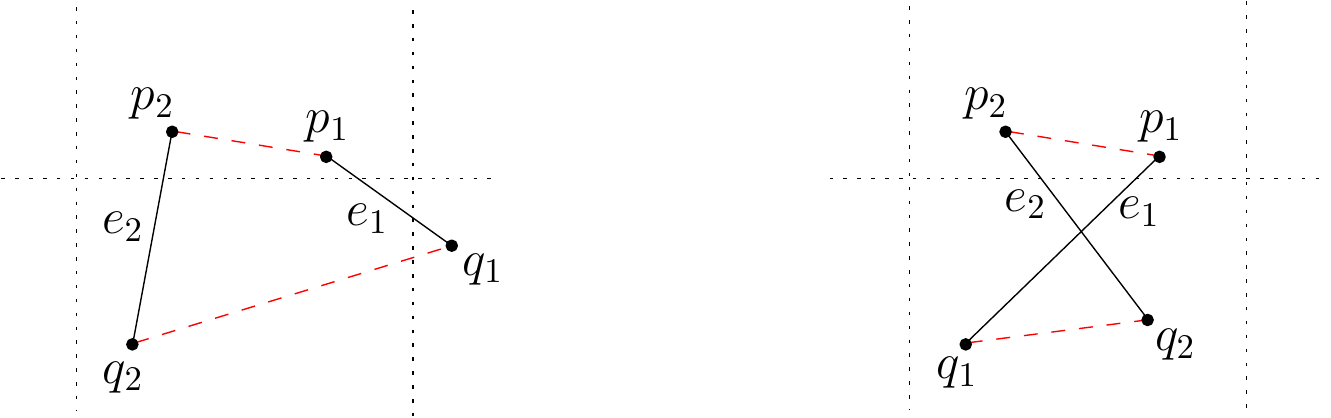}
        \caption{Rule~IV (left) and Rule~V (right).}
    \label{fig:ruleIV}
\end{figure}

Let $C(p)$ denote the cell containing point $p$.\\
{\bf Rule~IV} is applied to a d-edge $e_1=(p_1,q_1)$ and to a s-edge $e_2=(p_2,q_2)$, such that $C(p_1)=C(p_2)$ and $C(q_1)$ and $C(q_2)$ share a side, see Figure~\ref{fig:ruleIV}. $e_1$ and $e_2$ are replaced by an internal edge and a s-edge.\\  
{\bf Rule~V} is applied to a pair of s-edges $e_1=(p_1,q_1)$ and $e_2=(p_2,q_2)$, such that $C(p_1)=C(p_2)$ and $C(q_1)=C(q_2)$, see Figure~\ref{fig:ruleIV}. $e_1$ and $e_2$ are replaced by a pair of internal edges.\\

Let $M'_\times$ be the matching that is obtained after applying Rules~I-V. The following lemma summarizes some of the properties of $M'_\times$; its proof follows immediately from the discussion above.
\begin{lemma}
$M'_\times$ has the following properties:
\label{lem:properties}
\begin{enumerate}
\item
Each edge is either contained in a single cell, or connects between a pair of points in two adjacent cells.
\item
A corner has at most one d-edge associated with it.
\item
A d-edge is of length at most $\delta$.
\item
The two danger zones defined by a d-edge $e$ are empty of points of $P$.
\item
For each grid cell $C$, $deg(C) \le 4$ and these $deg(C)$ edges connect $C$ to $deg(C)$ of its adjacent cells.
\item
If $e$ is a d-edge in $M'_\times$ connecting between cells $C_1$ and $C_2$, and $C$ is a cell sharing a side with both $C_1$ and $C_2$, then there is no s-edge in $M'_\times$ connecting between $C$ and either $C_1$ or $C_2$.
\end{enumerate}
\end{lemma}

\subsection{Stage 2}

In this stage we construct $M_=$ according to $M'_\times$.
This stage consists of two parts.

\subsubsection{Part 1: Considering each cell separately}
In this part we consider each non-empty grid cell separately.
Let $C$ be a non-empty grid cell and set $P_C = P \cap C$.
We have to determine which of the points in $C$ will serve as endpoints of external edges.
The rest of the points will serve as endpoints of internal edges.
We have to consider both types of external edges, d-edges and s-edges.
We first consider the d-edges, then the s-edges, and, finally, after fixing the endpoints of the external edges, we form the internal edges.

%\begin{figure}[htp]
%    \centering
%        \includegraphics[width=.6\textwidth]{d-edge_endpoint}
%    \caption{Left: Picking the endpoint in $C$ of the corresponding d-edge in $M_=$. Right: $\Delta(p,\mbox{``up''})$. }
%    \label{fig:d-edge_endpoint}
%\end{figure}

For each corner $a$ of $C$ that has a d-edge (with endpoint in $C$) associated with it in $M'_\times$, consider the line through $a$ supporting $C$ and parallel to its appropriate diagonal, and pick the point $p_a$ in $C$ that is closest to this line as the endpoint (in $C$) of the corresponding d-edge in $M_=$.
%see Figure~\ref{fig:d-edge_endpoint} (left).
By Lemma~\ref{lem:properties} (Property~3), $P_C \cap D_\delta(a) \ne \emptyset$, and therefore the distance between $p_a$ and $a$ is less than $\sqrt{2}\delta$. Moreover, since the length of a side of $C$ is $2\sqrt{2}\delta$, each of the relevant corners is assigned a point of its own.
Observe also that a d-edge in $M_=$ with endpoint in $C$ will not cross any of the s-edges in $M_=$ with endpoint in $C$. This follows from Lemma~\ref{lem:properties} (Property~6).
We thus may ignore the points in $C$ that were chosen as endpoints of d-edges, and proceed to choose the endpoints of the s-edges.

\begin{figure}[htp]
    \centering
       \includegraphics[width=.80\textwidth]{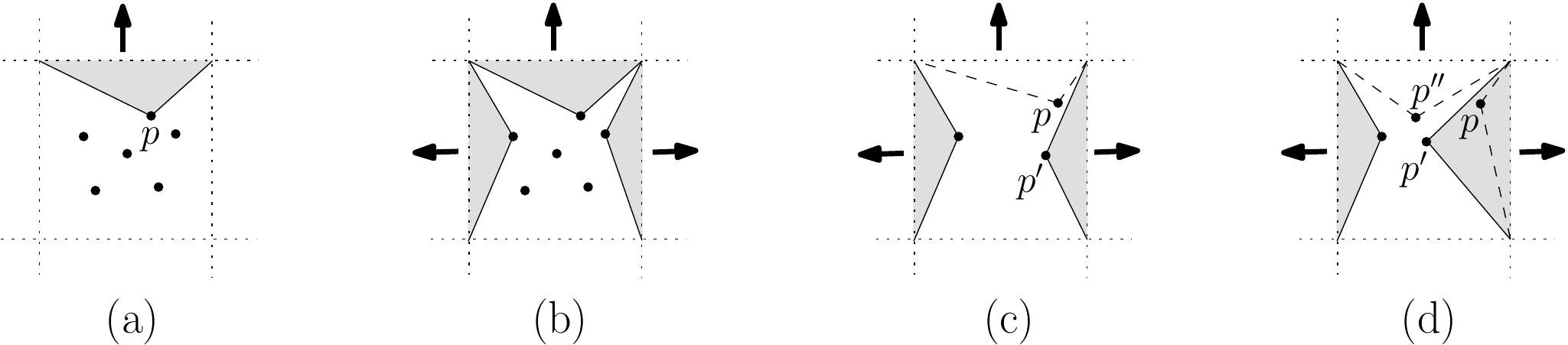}
   	\caption{(a) $\Delta(p,\mbox{``up''})$. Procedure~1: (b) $D=\{\mbox{``left''},\mbox{``up''},\mbox{``right''}\}$ and $|Q|=3$. (c) $p \not \in \Delta(p',\mbox{``right''})$ and $p$ is assigned to ``up''. (d) $p \in \Delta(p',\mbox{``right''})$ and $p$ is assigned to ``right''; $p''$ will be assigned to ``up''.}
   	\label{fig:proc1}
\end{figure}

Let $\Delta(p,dir)$ denote the triangle whose corners are $p$ and the two corners of $C$ in direction $dir$, where $dir \in \{\mbox{``up''},\mbox{``down''},\mbox{``left''},\mbox{``right''}\}$; see Figure~\ref{fig:proc1}(a).
If there is only one s-edge in $M'_\times$ with endpoint in $C$, and its direction is $dir$, then we pick the extreme point in direction $dir$ as the endpoint (in $C$) of the corresponding s-edge in $M_=$.
Assume now that there are $k$, $2 \le k \le 4$, s-edges in $M'_\times$ with endpoints in $C$.
We pick the $k$ endpoints (in $C$) of the corresponding s-edges in $M_=$ according to the recursive procedure below.

\floatname{algorithm}{Procedure}
				\begin{algorithm}
				\caption{PickEndpoints (a cell $C$, a set of directions $D=\{dir_1,\ldots,dir_k\}$)}
				\label{alg:Proc1}
				\begin{algorithmic}[1]
				
					\STATE $Q \leftarrow \emptyset$
					\FOR {each direction $dir_i \in D$}
					        \STATE let $p_i$ be the extreme point in direction $dir_i$
					 		\STATE $Q \leftarrow Q \cup \{p_i\}$
					\ENDFOR
					\IF {$|Q|=k$}
                            \STATE assign $p_i$ to direction $dir_i$, $i = 1,\ldots,k$
							\STATE Return\/$(Q)$
					
          \ENDIF
          \STATE let $p \in Q$ be a point that is the extreme point in two directions $dir_i, dir_j \in D$
					\STATE $Q' \leftarrow$ PickEndpoints\/$(C \setminus \{p\}, D \setminus \{dir_i\})$
					\STATE let $p' \in Q'$ be the point assigned to direction $dir_j$
					\IF {$p \not \in \Delta(p',dir_j)$}
						  	\STATE assign $p$ to direction $dir_i$	
					\ELSE
						 	\STATE $Q' \leftarrow$ PickEndpoints\/$(C \setminus \{p\}, D \setminus \{dir_j\})$
						 	\STATE assign $p$ to direction $dir_j$
					\ENDIF
          \STATE Return\/$(Q' \cup \{p\})$
\end{algorithmic}
\end{algorithm}

\begin{lemma}
The s-edges in $M_=$ with endpoints in $C$ do not cross each other.
\end{lemma}
\begin{proof}

By induction on $k$. If $k=1$, then there is nothing to prove. Assume $k \ge 2$ and
consider Procedure~1. If $|Q|=k$, then, for each $p_i \in Q$, the triangle $\Delta(p_i,dir_i)$ is empty, see Figure~\ref{fig:proc1}(b). Therefore, any two s-edges, one with endpoint $p_i$ and direction $dir_i$ and another with endpoint $p_j$ and direction $dir_j$, $i \ne j$, do not cross each other.
If $|Q| < k$, then, since $k \ge 2$, the directions $dir_i$ and $dir_j$ are not a pair of opposite directions.
Now, if $p$, the extreme point in directions $dir_i$ and $dir_j$, is not in $\Delta(p',dir_j)$ (Figure~\ref{fig:proc1}(c)), then $|Q' \cup \{p\}|=k$ and
each of the corresponding $k$ triangles is empty. If, however, $p \in \Delta(p',dir_j)$ (Figure~\ref{fig:proc1}(d)), then let $p''$ be the point assigned to direction $dir_i$ by the call to PickEndpoints in Line~14. We claim that $p \not \in \Delta(p'',dir_i)$. Indeed, if $p$ were in $\Delta(p'',dir_i)$, then either $p'' \in \Delta(p',dir_j)$ or $p' \in \Delta(p'',dir_i)$, contradicting in both cases the induction hypothesis for $k-1$.
\end{proof}

We are now ready to form the internal edges. Let $P_C^E \subseteq P_C$ be the set of points that were chosen as endpoints of external edges, and set $P_C^I = P_C \setminus P_C^E$. We show below that $P_C^I$ is contained in the interior of a convex region $R \subseteq C$, such that any external edge with endpoint in $C$ does not intersect the interior of $R$. Hence, if we form the internal edges by visiting the points in $P_C^I$ from left to right and matching each odd point with the next point in the sequence, then the resulting edges do not cross each other and do not cross any of the external edges.

It remains to define the convex region $R$. For each endpoint $p_i$ of a s-edge, draw a line $l_i$ through $p_i$ that is parallel to the side of $C$ crossed by the s-edge. Let $h_i$ be the half-plane defined by $l_i$ and not containing the s-edge, and set $R_i = h_i \cap C$. Similarly, for each endpoint $p_a$ of a d-edge, draw a line $l_a$ through $p_a$ that is parallel to the appropriate diagonal of $C$. Let $h_a$ be the half-plane defined by $l_a$ and not containing the d-edge, and set $R_a = h_a \cap C$. Finally, set $R = (\cap \{R_a\}) \cap (\cap \{h_i\})$. It is clear that $R$ is convex and that any external edge with endpoint in $C$ does not intersect the interior of $R$.
%We show that $P_C^I$ is contained in the interior of $R$. Let $p \in P_C^I$. Then, by the way we chose the endpoints of the d-edges,
%it is clear that $p \in \cap \{R_a\}$, and by the way we chose the endpoints of the s-edges, it is clear that $p \in \cap \{h_i\}$.
%We conclude that $p \in R$.
Moreover, by the way we chose the endpoints of the d-edges and s-edges, it is clear that $P_C^I$ is contained in the interior of $R$.

\subsubsection{Part 2: Putting everything together}

In this part we form the external edges of $M_=$.
For each external edge of $M'_\times$ connecting between cells $C_1$ and $C_2$, let $p$ (resp., $q$) be the point that was chosen as the endpoint in $C_1$ (resp., in $C_2$) of the corresponding edge of $M_=$, and match $p$ with $q$.

We have already shown that if $e_1$ and $e_2$ are two edges of $M_=$, for which there exists a grid cell containing an endpoint of both $e_1$ and $e_2$, then $e_1$ and $e_2$ do not cross each other. It remains to verify that a d-edge $e$ of $M_=$ connecting between $C_1$ and $C_2$ cannot cause any trouble in the cell $C$ through which it passes. Notice that $e \cap C$ is contained in the danger zone in $C$ defined by $e$ (or, more precisely, by the d-edge of $M'_\times$ corresponding to $e$). But, by Lemma~\ref{lem:properties} (Property~4), this danger zone is empty of points of $P$, thus $e$ cannot cross an internal edge contained in $C$. Moreover, Lemma~\ref{lem:properties} (Property~2) guarantees that there is no d-edge in $M_=$ crossing $e$, and Lemma~\ref{lem:properties} (Property~6) guarantees that there is no s-edge crossing $e$. We conclude that $e$ does not cause any trouble in $C$.

Finally, observe that $bn(M_=) \le 2\sqrt{10}\delta = 2\sqrt{10} \cdot bn(M_\times)$. This is true since the length of a d-edge in $M_=$ is at most $2\sqrt{2}\delta$ (i.e., the diagonal of $S_\delta(a)$), the length of an internal edge in $M_=$ is at most $4\delta$ (i.e., the diagonal of a single cell), and the length of a s-edge in $M_=$ is at most $2\sqrt{10}\delta$ (i.e., the diagonal of a pair of cells sharing a side).

The following theorem summarizes the main results of this section.
\begin{theorem}
Let $P$ be a set of $2n$ points in the plane.
Let $M_{\rm C}$ (resp., $M_{\rm NC}$) be a bottleneck matching (resp., a bottleneck non-crossing matching) of $P$. Then,
\begin{enumerate}
\item
$\frac{bn(M_{\rm NC})}{bn(M_{\rm C})} \le 2\sqrt{10}$.
\item
One can compute in $O(n^{1.5}\log^{0.5} n)$ time a non-crossing matching $M$ of $P$, such that $bn(M) \le 2\sqrt{10} \cdot bn(M_{\rm NC})$.
\end{enumerate}
\end{theorem}

\begin{proof}
{\it 1.} By applying the algorithm of this section to $M_{\rm C}$, we obtain a non-crossing matching $M$, such that
$bn(M_{\rm NC}) \le bn(M) \le 2\sqrt{10} \cdot bn(M_{\rm C})$.\\
{\it 2.} Compute $M_{\rm C}$ in $O(n^{1.5}\log^{0.5} n)$ time, using the algorithm of Chang et al.~\cite{CTL92}. Then, apply the algorithm of this section to $M_{\rm C}$ to obtain a non-crossing matching $M$, such that $bn(M) \le 2\sqrt{10} \cdot bn(M_{\rm C}) \le 2\sqrt{10} \cdot bn(M_{\rm NC})$. It is easy to see that the time complexity of the latter stage is only $O(n\log n)$.
\end{proof}

{\bf Remarks.}
{\bf 1.} We can improve the approximation ratio to $(1+\sqrt{2})\sqrt{5}$ by reducing the cell size to $1+\sqrt{2}$, and picking the endpoints of the d-edges more carefully.\\
{\bf 2.} There exists a set $P$ of $2n$ points in the plane, for which $bn(M_{\rm NC}) \ge \frac{\sqrt{85}}{8} \cdot bn(M_{\rm C})$; see Figure~\ref{fig:l_bound}.\\
{\bf 3.} It is interesting to note that in the bipartite version, the ratio $bn(M_{\rm NC})/bn(M_{\rm C})$ can be linear in $n$, even if the red and blue points are separated by a line; see Figures~\ref{fig:conc-2}.

\begin{figure}[htp]
    \centering
       \includegraphics[width=.4\textwidth]{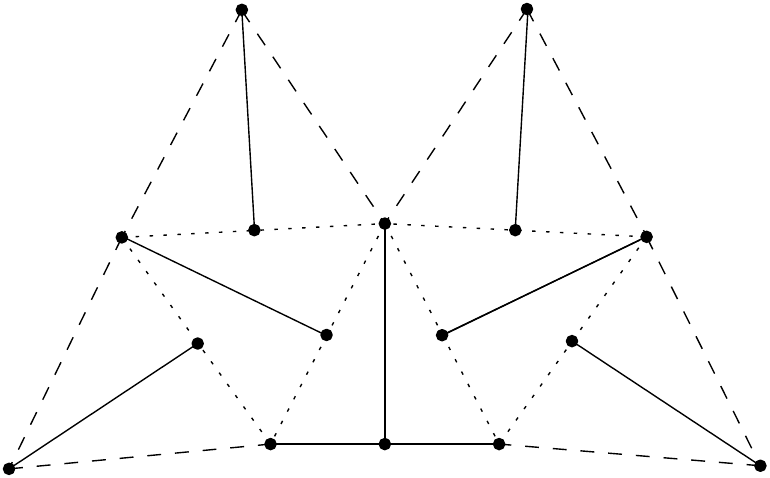}
   	\caption{The length of the solid edges is 1 and the length of the dashed edges is $\sqrt{85}/8$. The bottleneck (crossing) matching consists of the 8 solid edges, while any bottleneck non-crossing matching must use at least one of the dashed edges.}
   	\label{fig:l_bound}
\end{figure}

\begin{figure}[htp]
    \centering
        \includegraphics[width=.77\textwidth]{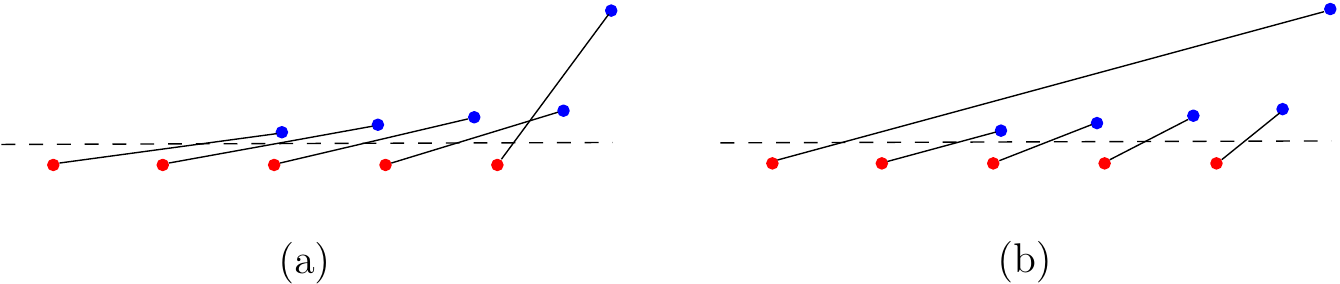}
    \caption{(a) Crossing and (b) non-crossing matching of $n$ blue and $n$ red points separated by a line.}
    \label{fig:conc-2}
\end{figure}

%%%%%%%%%%%%%%%%%%%%%%%%%%%%%%%%%%%%%%%%%%%% Section 4 %%%%%%%%%%%%%%%%%%%%%%%%%%%%%%%%%%%%%%%
%%%%%%%%%%%%%%%%%%%%%%%%%%%%%%%%%%%%%%%%%%%%%%%%%%%%%%%%%%%%%%%%%%%%%%%%%%%%%%%%%%%%%%%%%%%%%%
\section{Solving Special Cases}\label{sec:Sec4}
In this section, we consider two special cases of the BNCM problem: when the points of $P$ are in convex position, i.e., the points in $P$ form the vertices of a convex polygon, and when the points are located on a circle.

%%%%%%%%%%%%%%%%%%%%%%%%%%%%%%%%%%%%%%%%%%%%%%%%%%%%%%%%%%%%%%%%%%%%%%%%%%%%%%%%%%%%%%%%%%%%%

\subsection{Matching Points in Convex Position}
Let $\{p_1,p_2,\ldots,p_{2n}\}$ denote the vertices of the convex polygon, that is obtained by connecting the points in $P$, ordered in clockwise-order with an arbitrary first point $p_1$; see Figure~\ref{fig:convex-1}. Notice that, a non-crossing perfect matching in $\{p_1,p_2,\ldots,p_{2n}\}$ always exists. Let $M^*$ be an optimal matching of $P$, i.e., a non-crossing matching with minimum bottleneck. We first observe that, for each edge $(p_i,p_j)$ in $M^*$, $i+j$ is odd. According to this observation, we define the following weight function
\[
w_{i,j}=
\begin{cases}
|p_ip_j| &\text{: if $i+j$ is odd;}\\
\infty &\text{: otherwise.}\\
\end{cases}
\]

\begin{figure}[htp]
    \centering
        \includegraphics[width=.50\textwidth]{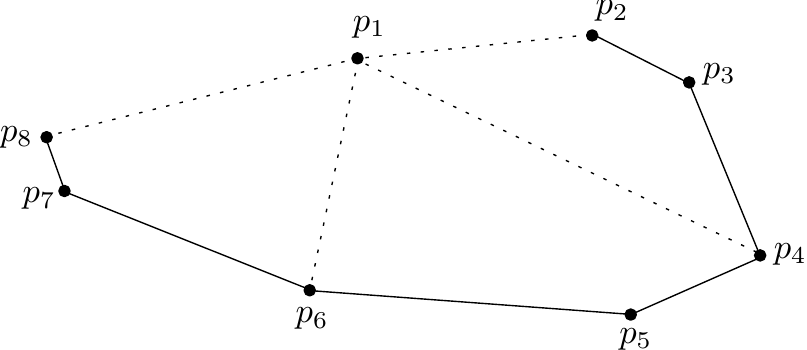}
    \caption{The convex polygon that is obtained from $P$, for $n=4$. $p_1$ can be only matched to the points $p_2, p_4, p_6$ or $p_8$.}
    \label{fig:convex-1}
\end{figure}

Notice also that, for any $1\le i<j\le 2n$, since the points in $\{p_1,p_2,\ldots,p_{2n}\}$ are in convex position, so are the points in $\{p_i,p_{i+1},\ldots,p_{j}\}$. Let $M[i,j]$ denote the bottleneck of an optimal matching between the points in $\{p_i,p_{i+1},\ldots,p_{j}\}$. Since $M^*$ is a perfect matching, each point, particularly $p_1$, is matched in $M^*$. Let $p_k$ (for an even $k$) be the point that is matched to $p_1$ in $M^*$. Hence, the bottleneck of $M^*$ (i.e., $M[1,2n]$) is equal to $\max\{w_{1,k},M[2,k-1],M[k+1,2n]\}$. Thus, in order to compute $M[1,2n]$, we compute $\max\{w_{1,k},M[2,k-1],M[k+1,2n]\}$ for each even $k$ between 2 and $2n$, and we take the minimum over these values. In general, for every $1\le i<j\le 2n$, we have
\[
M[i,j]=\min_{k=i+1,i+3,\ldots,j}
\begin{cases}
w_{i,k} &\text{: if $k=i+1=j$;}\\
\max\{w_{i,k},M[k+1,j]\} &\text{: if $k=i+1$;}\\
\max\{w_{i,k},M[i+1,k-1]\} &\text{: if $k=j$;}\\
\max\{w_{i,k},M[k+1,j],M[i+1,k-1]\} &\text{: otherwise.}
\end{cases}
\]

We can compute $M[1,2n]$ using dynamic programming. The dynamic programming table $M$ has $n$ rows and $n$ columns. Each cell $M[i,j]$ corresponds to a solution of the problem for the set $\{p_i,p_{i+1},\ldots,p_j\}$, and it can be computed by at most $n$ lookups in the table. We fill the table iteratively in such a way insuring that, for each cell $M[i,j]$, all the values needed to compute $M[i,j]$ are already computed. To do that, we first fill the diagonal $M[1,2], M[2,3],\ldots,M[{n-1},n]$, and then we fill the above diagonals one by one. This leads to solve the problem in $O(n^3)$ time and in $O(n^2)$ space.

%\begin{figure}[htp]
%    \centering
%        \includegraphics[width=.48\textwidth, height=0.38\textwidth]{convex-2}
%    \caption{Filling the table iteratively.}
%    \label{fig:convex-2}
%\end{figure}
The following theorem summarizes this result.
\begin{theorem}
Given a set $P$ of $2n$ points in convex position, one can compute a bottleneck non-crossing matching in time $O(n^3)$ and in space $O(n^2)$.
\end{theorem}

\paragraph{Remark.} The same result holds for computing a bottleneck non-crossing matching of two (red and blue) sets of points in convex position, and of a set of points on the boundary of a simple polygon, with the constraint that edges must not leave the polygon.

%%%%%%%%%%%%%%%%%%%%%%%%%%%%%%%%%%%%%%%%%%%%%%%%%%%%%%%%%%%%%%%%%%%%%%%%%%%%%%%%%%%%%%%%%%%%%%
\subsection{Matching Points on a Circle}
Let $P=\{p_1,p_2,\ldots,p_{2n}\}$ be a set of points located (in clockwise-order) on a boundary of a circle. In this section, we show how to find a bottleneck non-crossing matching of $P$ in time $O(n)$.

Let $M^*$ be an optimal matching of $P$, i.e., a non-crossing matching with minimum bottleneck. We first claim that each edge in $M^*$ connects between two consecutive points from $P$, i.e., if $(p_i,p_j)$ in $M^*$ then $j-i=\pm1$ (assume that $p_1=p_{2n+1}$ and $p_{2n}=p_0$). To see that, assume that there is an edge $(p_i,p_j)$ in $M^*$ that violates the claim, and consider the line that passes through $p_i$ and $p_j$. It divides the circle into two arcs $\wideparen{p_ip_j}$ and $\wideparen{p_jp_i}$. Assume without loss of generality that $\wideparen{p_ip_j}$ is shorter than $\wideparen{p_jp_i}$. Then, it is clear that the distance between any two points on the arc $\wideparen{p_ip_j}$ is at most $p_ip_j$, and, since there is an even number of points on $\wideparen{p_ip_j}$, we can match the points $p_i, p_{i+1}, \cdots, p_{j-1}, p_j$ consecutively without affecting the bottleneck of $M^*$.

Therefore, since there are only two ways to match the points consecutively, we can compute an optimal matching of $P$ in linear-time.
\begin{theorem}
Given a set $P$ of $2n$ points on a circle, one can compute a bottleneck non-crossing matching in time $O(n)$.
\end{theorem}

\old{

%%%%%%%%%%%%%%%%%%%%%%%%%%%%%%%%%%%%%%%%%%%% Section 5 %%%%%%%%%%%%%%%%%%%%%%%%%%%%%%%%%%%%%%%
%%%%%%%%%%%%%%%%%%%%%%%%%%%%%%%%%%%%%%%%%%%%%%%%%%%%%%%%%%%%%%%%%%%%%%%%%%%%%%%%%%%%%%%%%%%%%%
\section{Concluding Remarks}\label{sec:Sec5}

We studied the problem of finding a bottleneck non-crossing matching of a set of $2n$ points in the plane. We showed that the problem cannot be approximated within a factor of $3\sqrt{2}/4$, and we gave a polynomial-time approximation algorithm with approximation ratio $2\sqrt{5}$. In addition, we showed that the problem can be solved in $O(n^3)$ time when the points in convex position, and in $O(n)$ time when the points on a circle. It will be interesting to achieve an approximation ratio better than $2\sqrt{5}$ and/or to improve the running time of exact algorithm for the convex position case.

During this work, we considered the bipartite version of the problem, i.e., matching $n$ blue points to $n$ red points. Unlike the general version, we observed that the ratio between bottleneck of the crossing and the non-crossing matchings can be linear as shown in Figure~\ref{fig:conc-1}, even in the restricted case when the blue and the red points are separated by a line, as shown in Figure~\ref{fig:conc-2}. This implies that our algorithm for the general version can not be extended to give a ``good'' result for this version.
\begin{figure}[htp]
    \centering
        \includegraphics[width=.89\textwidth]{conc-1}
    \caption{(a) Crossing and (b) non-crossing matchings of $n$ blue and $n$ red points.}
    \label{fig:conc-1}
\end{figure}
\begin{figure}[htp]
    \centering
        \includegraphics[width=.89\textwidth]{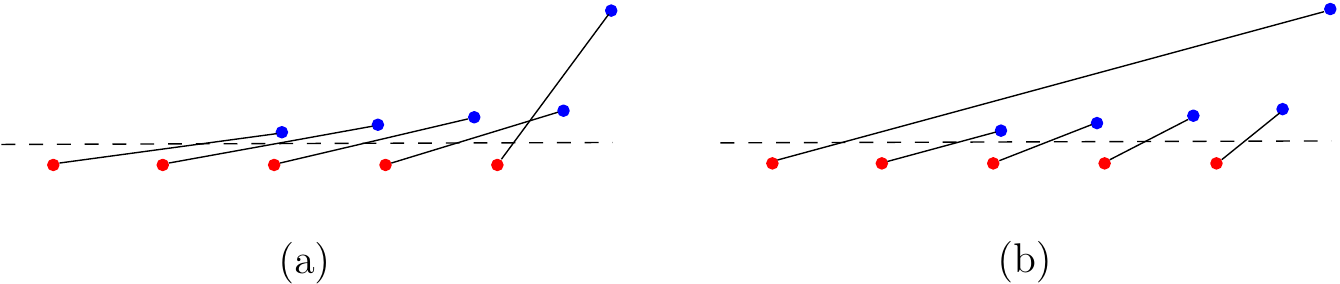}
    \caption{The restricted case: (a) Crossing and (b) non-crossing matchings of $n$ blue and $n$ red points separated by a line.}
    \label{fig:conc-2}
\end{figure}

It is clear that the hardness proof of the general version holds also for the bipartite version but not for the restricted case. Interesting questions are to provide an approximation algorithm for this version and to prove hardness of the restricted case.
}

%\newpage
\bibliographystyle{plain}
\bibliography{ref}

\end{document}